\documentclass[runningheads,a4paper]{llncs}

\usepackage{algorithm, algorithmic}
\usepackage{array}
\usepackage{longtable}
\usepackage{caption}
\usepackage{url}
\usepackage{graphicx}
\usepackage{tikz}
\usepackage{mathtools}
\DeclarePairedDelimiter{\ceil}{\lceil}{\rceil}

\renewcommand{\ceil}[1]{\left\lceil{#1}\right\rceil}

\newcommand{\set}[1]{\left\{#1\right\}}
\newcommand{\abs}[1]{\left|#1\right|}
\newcommand{\paren}[1]{\left(#1\right)}

\newcommand{\NC}{Network Construction}
\newcommand{\NCOC}{Network Construction with Ordered Constraints}

\newcommand{\keywords}[1]{\par\addvspace\baselineskip\noindent\keywordname\enspace\ignorespaces#1}

\begin{document}

\mainmatter  

\title{Network Construction with Ordered Constraints}
\titlerunning{Network Construction with Ordered Constraints}

\author{Yi Huang \and Mano Vikash Janardhanan \and Lev Reyzin\thanks{Supported in part by ARO grant 66497-NS.}}
\authorrunning{Yi Huang \and Mano Vikash Janardhanan \and Lev Reyzin}

\institute{University of Illinois at Chicago,\\
\texttt{\{yhuang89, mjanar2, lreyzin\}@uic.edu}\\
\url{}}

\maketitle

\captionsetup[tlongable]{labelfont=bf}
\captionsetup[table]{labelfont=bf}

\begin{abstract}
In this paper, we study the problem of constructing a network by observing ordered connectivity constraints, which we define herein. These ordered constraints are made to capture realistic properties of real-world problems that are not reflected in previous, more general models. We give hardness of approximation results and nearly-matching upper bounds for the offline problem, and we study the online problem in both general graphs and restricted sub-classes. In the online problem, for general graphs, we give exponentially better upper bounds than exist for algorithms for general connectivity problems. For the restricted classes of stars and paths we are able to find algorithms with optimal competitive ratios, the latter of which involve analysis using a potential function defined over pq-trees.
\end{abstract}

\section{Introduction}\label{sec:introduction}

In this paper, we study the problem of recovering a network after observing how information propagates through the network. Consider how a tweet (through ``retweeting'' or via other means) propagates through the Twitter network -- we can observe the identities of the people who have retweeted it and the timestamps when they did so, but may not know, for a fixed user, via whom he got the original tweet. So we see a chain of users for a given tweet. This chain is semi-ordered in the sense that, each user retweets from some one before him in the chain, but not necessarily the one directly before him. Similarly, when a virus such as Ebola spreads, each new patient in an outbreak is infected from someone who has previously been infected, but it is often not immediately clear from whom.

In a graphical social network model with nodes representing users and edges representing links, an ``outbreak'' illustrated above is captured exactly by the concept of an \textbf{ordered constraint} which we will define formally below. One could hope to be able to learn something about the structure of the network by observing repeated outbreaks, or a sequence of ordered constraints.

Formally we call our problem \textbf{\NCOC} and define it as follows. Let $V=\set{v_1,\dots,v_n}$ be a set of vertices. An \textbf{ordered constraint} $\mathcal{O}$ is an ordering on a subset of $V$ of size $s\geq2$. The constraint $\mathcal{O}=\paren{v_{k_1},\dots, v_{k_{s}}}$ is satisfied if for any $2\leq i\leq s$, there exists at least one $1\leq j<i$ such that the edge $e=\set{v_{k_j}, v_{k_i}}$ is included in a solution. Given a collection of ordered constraints $\set{\mathcal{O}_1,\dots,\mathcal{O}_r}$, the task is to construct a set  $E$ of edges among the vertices $V$ such that all the ordered constraints are satisfied and $|E|$ is minimized.

We can see that our newly defined problem resides in a middle ground between path constraints, which are too rigid to be very interesting, and the well-studied subgraph connectivity constraints~\cite{angluinAR2015network,korachS2003clustering,korachS2008complete}, which are more relaxed. 
The established subgraph connectivity constraints problem involves getting an arbitrary collection of connectivity constraints $\set{S_1,\dots, S_{r}}$ where each $S_i\subset V$ and requires vertices in a given constraint to form a connected induced subgraph. The task is to construct a set  $E$ of edges satisfying the connectivity constraints such that $|E|$ is minimized.

We want to point out one key observation relating the ordered constraint to the connectivity constraint -- an ordered constraint $\mathcal{O}=\paren{v_{k_1},\dots, v_{k_s}}$ is equivalent to $s-1$ connectivity constraints $S_2,\dots, S_{s}$, where $S_i=\set{v_{k_1},\dots, v_{k_i}}$. We note that this observation plays an important role in several proofs in this paper which employ previous results on subgraph connectivity constraints -- in particular, upper bounds from the more general case can be used in the ordered case (with some overhead), and our lower bounds apply to the general problem.

In the offline version of the {\NCOC} problem, the algorithm is given all of the constraints all at once; in the \textbf{online} version of the problem, the constraints are given one by one to the algorithm, and edges must be added to satisfy each new constraint when it is given.  Edges cannot be removed.

An algorithm is said to be $\mathbf{c}$-\textbf{competitive} if the cost of its solution is less than $c$ times OPT, where OPT is the best solution in hindsight ($c$ is also called the competitive ratio).
When we restrict the underlying graph in a problem to be a class of graphs, e.g.~trees, we mean all the constraints can be satisfied, in an optimal solution (for the online case, in hindsight), by a graph from that class. 

\subsection{Past Work}\label{subsec:pastWork}

In this paper we study the problem of network construction from ordered constraints. This is an extension of the more general model where constraints come unordered.

For the general problem, Korach and Stern~\cite{korachS2003clustering} had some of the initial results, in particular for the case where the constraints can be optimally satisfied by a tree, they give a polynomial time algorithm that finds the optimal solution. In subsequent work, in~\cite{korachS2008complete} Korach and Stern considered this problem for the even more restricted problem where the optimal solution forms a tree, and all of the connectivity constraints must be satisfied by stars.

Then, Angluin~et~al.~\cite{angluinAR2015network} studied the general problem, where there is no restriction on structure of the optimal solution, in both the offline and online settings. In the offline case, they gave nearly matching upper and lower bounds on the hardness of approximation for the problem. In the online case, they give a $O(n^{2/3}\log^{2/3}n)$-competitive algorithm against oblivious adversaries; we show that this bound can be drastically improved in the ordered version of the problem. They also characterized special classes of graphs, i.e.~stars and paths, which we are also able to do herein for the ordered constraint case. Independently of that work, Chockler~et~al.~\cite{chocklerMTV2007constructing} also nearly characterized the offline general case.

In a different line of work Alon~et~al.~\cite{alonAABN2006general} explore a wide range of network optimization problems; one problem they study involves ensuring that a network with fractional edge weights has a flow of $1$ over cuts specified by the constraints.  Alon~et~al.~\cite{alonAA2003online} also study approximation algorithms for the Online Set Cover problem which have been shown by Angluin~et~al~\cite{angluinAR2015network} to have connections with Network Construction problems.

In related areas, Gupta~et~al.~\cite{guptaKR2012online} considered a network design problem for pairwise vertex connectivity constraints.
Moulin and Laigret~\cite{MoulinL11} studied network connectivity constraints from an economics perspective. Another motivation for studying this problem is to discover social networks from observations. This and similar problems have also been studied in the learning context~\cite{angluinAR2008optimally,angluinAR2010inferring,gomez-RodriguezLK2012inferring,saitoNK2008prediction}. 

Finally, in query learning, the problem of discovering networks from connectivity queries has been much
studied~\cite{alonA2005learning,alonBKRS2004learning,angluinC2008learning,beigelAKAF2001optimal,grebinskiK1998reconstructing,reyzinS2007learning}. In active learning of hidden networks, the object of the algorithm is to learn the network exactly. Our model is similar, except the algorithm only has the constraints it is given, and the task is to output the cheapest network consistent with the constraints.

\subsection{Our results}\label{subsec:result}
In Section~\ref{sec:offline}, we examine the offline problem, and show that the {\NC} problem is NP-Hard to approximate within a factor of $\Omega(\log n)$. A nearly matching upper bound comes from {Theorem 2} of \cite{angluinAR2015network}.

In Section~\ref{sec:online}, we study online problem. For problems on $n$ nodes, for $r$ constraints, we give an $O\paren{(\log r+\log n)\log n}$ competitive algorithm against oblivious adversaries, and an $\Omega(\log n)$ lower bound (Section~\ref{subsec:onlineGeneral}). 

Then, for the special cases of stars and paths (Sections~\ref{subsubsec:onlineStars} and~\ref{subsubsec:paths}), we find asymptotic optimal competitive ratios of $3/2$ and $2$, respectively. The proof of the latter uses a detailed analysis involving pq-trees~\cite{boothL1976testing}. The competitive ratios are asymptotic in $n$.

\section{The offline problem}\label{sec:offline}
In this section, we examine the {\NCOC} problem in the offline case.  We are able to obtain the same lower bound as Angluin~et~al.~\cite{angluinAR2015network} in the general connectivity constraints case.
\begin{theorem}
\label{thm:offlineApprox}
If P$\neq$NP, the approximation ratio of the \textbf{\NCOC} problem is $\Omega{(\log {n})}$.
\end{theorem}
\begin{proof}
We prove the theorem by reducing from the Hitting Set problem. Let $(U,\mathcal{S})$ be a hitting set instance, where $U=\set{u_1,\dots,u_n}$ is the universe, and $\mathcal{S}=\set{S_1,\dots,S_m}$ is a set of subsets of $U$. A subset $H\subset {U}$ is called a hitting set if $H\cap S_i\neq\emptyset$. The objective of the Hitting Set problem is to minimize $\abs{H}$. We know from \cite{feigeU1998threshold,razS1997sub} that the Hitting Set problem cannot be approximated by any polynomial time algorithm within a ratio of $o(\log{n})$ unless P$=$NP. Here we show that the {\NC} problem is inapproximable better than an $O(\log n)$ factor by first showing that we can construct a corresponding {\NC} instance to any given Hitting Set instance, and then showing that if there is a polynomial time algorithm that can achieve an approximation ratio $o(\log{n})$ to the {\NC} problem, then the Hitting Set problem can also be approximated within in a ratio of $o(\log{n})$, which is a contradiction.

We first define a {\NC} instance, corresponding to a given Hitting Set instance $(U, \mathcal{S})$, with vertex set $U \cup W$, where $W=\set{w_1,\dots,w_{n^c}}$ for some $c>2$. Note that we use the elements of the universe of hitting set instance as a part of the vertex set of {\NC} instance. The ordered constraints are the union of the following two sets:
\begin{itemize}
\item $\set{(u_i, u_j)}_{1\leq i<j\leq n}$;
\item $\set{(S_k, w_l)}_{S_k\in\mathcal{S}, 1\leq l\leq n^c}$,
\end{itemize}
where by $(S_k, w_l)$ we mean an ordered constraint with all vertices except the last one from a subset $S_k$ of $U$, while the last vertex $w_l$ is an element in $W$. The vertices from $S_k$ are ordered arbitrarily. 

We note that the first set of ordered constrains forces a complete graph on $U$, and the second set of ordering demands that there is at least one edge going out from each $S_k$ connecting each element in $W$. More specifically let $E_l$ denote the set of edges incident to $w_l$ belonging to any solution to the {\NC} instance. Because of the second set of ordered constraints, the set $H_l=\set{u\in U|\set{u,w_l}\in E_l}$ is a hitting set of $\mathcal{S}$! 

Let $H\subset U$ be any optimal solution to the hitting set instance, and denote by $\mathrm{OPT}_\mathrm{H}$ the size of $H$, it is easy to see the two sets of ordered constraints can be satisfied by putting a complete graph on $U$ and a complete bipartite graph between $H$ and $W$. Hence the optimal solution to the {\NC} instance satisfies 
\[
	\mathrm{OPT}\leq\binom{n}{2}+n^c\ \mathrm{OPT}_\mathrm{H},
\]
where $\mathrm{OPT}$ is the minimum number of edges needed to solve the {\NC} instance. Let us assume that there is a polynomial time approximation algorithm to the {\NC} problem that adds $\mathrm{ALG}$ edges. Without loss of generality we can assume that the algorithm add no edge among vertices in $W$, because any edge within $W$ can be removed without affecting the correctness of the solution, which implies that $\mathrm{ALG}=\binom{n}{2}+\sum_{l}^{n^c}|E_l|$. Now if $\mathrm{ALG}$ is in the order $o\paren{\log{n}\mathrm{OPT}}$, from the fact that $|H_l|=|E_l|$, we get
\begin{eqnarray*}
\min_{1\leq l\leq n^c}|H_l|\leq \frac{\mathrm{ALG}-\binom{n}{2}}{n^{c}}&=&\frac{{o\paren{\log{n}\paren{\binom{n}{2}+n^c\ \mathrm{OPT}_\mathrm{H}}}-\binom{n}{2}}}{n^c}\\
&=&o\paren{\log{n}\ \mathrm{OPT}_\mathrm{H}},
\end{eqnarray*}
which means by finding the smallest set $H_{l_0}$ among all the $H_l$s, we get a hitting set that has size within an $o(\log{n})$ factor of the optimal solution to the Hitting Set instance, which is a contradiction.
\qed
\end{proof}

We also observe that the upper bound from the more general problem implies a bound in our ordered case.
We note the upper and lower bounds match when $r = poly(n)$.

\begin{corollary}[of Theorem 2 from Angluin~et~al.~\cite{angluinAR2015network}]
There is a polynomial time $O(\log r+\log n)$-approximation algorithm for the \NCOC problem on $n$ nodes and $r$ constraints.
\end{corollary}
\begin{proof}
Observing that $r$ ordered constraints imply at most $nr$ unordered constraints on a graph with $n$ nodes, we can use the $O(\log r)$ upper bound from Angluin~et~al.~\cite{angluinAR2015network}.
\qed
\end{proof}

\section{The online problem}\label{sec:online}

Here, we study the online problem, where constraints come in one at a time, and the algorithm must satisfy them by adding edges as the constraints arrive.

\subsection{Arbitrary graphs}\label{subsec:onlineGeneral}
\begin{theorem}\label{thm:onlineGeneral}
The competitive ratio for \textbf{Online \NCOC} problem on $n$ nodes and $r$ ordered constraints has an upper bound of $O\paren{(\log r+\log n)\log n}$ against an oblivious adversary.
\end{theorem}
\begin{proof}

To prove the statement, we first define the \textbf{Fractional Network Construction} problem, which has been shown by Angluin~et~al.~\cite{angluinAR2015network} to have an $O(\log n)$-approximation algorithm. The upper bound is then obtained by applying a probabilistic rounding scheme to the fractional solution given by the approximation. The proof heavily relies on arguments developed by Buchbinder and Noar~\cite{buchbinderN2009design}, and Angluin~et~al.~\cite{angluinAR2015network}.

In the \textbf{Fractional Network Construction} problem, we are also given a set of vertices and a set of constraints $\set{S_1,\dots,S_r}$ where each $S_i$ is a subset of the vertex set. Our task is to assign weights $w_e$ to each edge $e$ so that the maximum flow between each pair of vertices in $S_i$ is at least $1$. The optimization problem is to minimize $\sum w_e$. Since subgraph connectivity constraint is equivalent to requiring a maximum flow of $1$ between each pair of vertices with edge weight $w_e\in\set{0,1}$, the fractional network construction problem is the linear relaxation of the subgraph connectivity problem. Lemma~2 of Angluin~et~al.~\cite{angluinAR2015network} gives an algorithm that multiplicatively updates the edge weights until all the flow constraints are satisfied. It also shows that the sum of weights given by the algorithm is upper bounded by $O(\log n)$ times the optimum. 

As we pointed out in the introduction, an ordered constraint $\mathcal{O}$ is equivalent to a sequence of subgraph connectivity constraints. So in the first step, we feed the $r$ sequences of connectivity constraints, each one is equivalent to an ordered constraint, to the approximation algorithm to the fractional network construction problem and get the edge weights. Then we apply a rounding scheme similar to the one considered by Buchbinder and Noar~\cite{buchbinderN2009design} to the weights. For each edge $e$, we choose $t$ random variables $X(e, i)$ independently and uniformly from $[0,1]$, and let the threshold $T(e)=\min_{i=1}^{t}X(e,i)$. We add $e$ to the graph if $w_e\geq T(e)$.

Since the rounding scheme has no guarantee to produce a feasible solution, the first thing we need to do is to determine how large $t$ should be to make all the ordered constraints satisfied with high probability.

We note that an ordered constraint $\mathcal{O}_i=\set{v_{i1}, v_{i2},\dots,v_{i{s_i}}}$ is satisfied if and only if the $(s-1)$ connectivity constraints $\set{v_{i1},v_{i2}}$, $\dots$, $\set{v_{i1},\dots, v_{i{s_i-1}},v_{i{s_i}}}$ are satisfied which is equivalent, in turn, to the fact that there is an edge that goes across the $\paren{\set{v_{i1},\dots, v_{i{j-1}}}, \set{v_{ij}}}$ cut, for $2\leq j\leq s_i$. For any fixed cut $C$, the probability the cut is not crossed equals $\prod_{e\in C}(1-w_e)^{t}\leq \exp\paren{-t\sum_{e\in C}w_e}.$ By the max-flow min-cut correspondence, we know that $\sum_{c\in C}w_e\geq 1$ in the fractional solution given by the approximation algorithm for all cut $C=\paren{\set{v_{i1},\dots, v_{i{j-1}}}, \set{v_{ij}}}$, $1\leq i \leq r$, $2\leq j\leq s_i$, and hence the probability that there exists at least one unsatisfied $\mathcal{O}_i$ is upper bounded by $rn\exp\paren{-t}$. So $t=c(\log n+\log r)$, for any $c>1$, makes the probability that the rounding scheme fails to produce a feasible solution approaches $0$ as $n$ increases.  

Because the probability that $e$ is added equals the probability that at least one $X(e,i)$ is less than $w_e$, and hence is upper bounded by $w_et$, we get the expected number of edges added is upper bounded by $t\sum{w_e}$ by linearity of expectation. Since the fractional solution is upper bounded by $O(\log n)$ times the optimum of the fractional problem, which is upper bounded by any integral solution, our rounding scheme gives a solution that is $O\paren{(\log r+\log n)\log n}$ times the optimum. 
\qed
\end{proof}

\begin{corollary}
If the number of ordered constraints $r=\mathrm{poly}(n)$, then the algorithm above gives a $O\paren{(\log n)^2}$ upper bound for the competitive ratio against an oblivious adversary.
\end{corollary}

\begin{remark}
We can generalise theorem \ref{thm:onlineGeneral} to the weighted version of the Online {\NCOC}  problem. In the weighted version, each edge $e=(u,v)$ is associated with a cost $c_e$ and the task is to select edges such that the connectivity constraints are satisfied and $\sum c_e w_e$ is minimised where $w_e\in\{0,1\}$ is a variable indicating whether an edge is picked or not and $c_e$ is the cost of the edge. The same approach in the proof of Theorem~\ref{thm:onlineGeneral} gives an upper bound of $O\paren{(\log r+\log n)\log n}$ for the competitive ratio of the weighted version of the Online {\NCOC} problem.
\end{remark}

\begin{theorem}
This is a $\Omega(\log n)$ lower bound for the competitive ratio for the \textbf{Online \NCOC} problem against an oblivious adversary.
\end{theorem}

\begin{proof}
The adversary divides the vertex set into two parts $U$ and $V$, where $|U|=\sqrt{n}$ and $|V|=n-\sqrt{n}$, and gives the constraints as follows. Firstly, it forces a complete graph in $U$ by giving the constraint $\set{u_i,u_j}$ for each pair of vertices $u_i,u_j\in U$. At this stage both the algorithm and optimal solution will have a clique in $U$, which costs $\Theta(n)$.

Then, for each $v\in V$, first fix a random permutation $\pi_v$ on $U$ and give the ordered constraint
\[
\mathcal{O}_{(v,i)}=\paren{\pi_v(1),\pi_v(2),\ldots,\pi_v(i), v}.
\]
First note that all these constraints can be satisfied by adding $e_v=\set{\pi_v(1), v}$ for each $v\in V$ which costs $\Theta(n)$. However, the adversary gives constraints in the following order:  
\[
\mathcal{O}_{(v,\sqrt{n})},\mathcal{O}_{(v,\sqrt{n}-1)},\dots,\mathcal{O}_{(v,1)}.
\]
We now claim that for each $v\in V_2$, the algorithm will add $\Omega(\log n)$ edges in expectation. This is because each edge added by the algorithm is a random guess for $\pi_v(1)$ and this edge cuts down the number of unsatisfied $\mathcal{O}_{(v,i)}$ by half in expectation. This means the algorithm adds $\Omega(n+n\log n)$ edges. This gives us the desired result because $\mathrm{OPT}=O(n)$.
\qed
\end{proof}

Now we study the online problem when it is known that an optimal graph can be a star or a path.  These special cases are challenging in their own right and are
often studied in the literature to develop more general techniques~\cite{angluinAR2015network}.

\subsection{Stars}\label{subsubsec:onlineStars}
\begin{theorem}\label{thm:onlineStarLowerAndUpper}
The optimal competitive ratio for the \textbf{Online \NCOC} problem when the algorithm knows that an optimal solution forms a \textbf{star} is asymptotically $3/2$.
\end{theorem}
\begin{proof}
For lower bound, we note that the adversary can simply give $\mathcal{O}_{i}=\paren{v_1,v_2,v_i}$, $3\leq i\leq n$ obliviously for the first $n-2$ rounds. Then an algorithm, besides adding $\set{v_1,v_2}$ in the first round, can only choose from adding either $\set{v_1,v_i}$ or $\set{v_2, v_i}$, or both in each round. After the first $n-2$ rounds, the adversary counts the number of $v_1$ and $v_2$'s neighbors, and chooses the one with fewer neighbors, say $v_1$, to be the center by adding $(v_1, v_i)$ for some $3\leq i\leq n$. Since the algorithm has to add at least $\ceil{(n-2)/2}$ edges that are unnecessary in the hindsight, we get an asymptotic lower bound $3/2$. 

For upper bound, assume that the first ordered constraint is $\mathcal{O}_1$ is $(v_1, v_2,\dots)$, the algorithm works as follows:
\begin{enumerate}
\item It adds $\set{v_1,v_2}$ in the first round.
\item Then for any constraint that starts with $v_1$ and $v_2$, it splits the remaining vertices in the constraint (other than  $v_1$ and $v_2$) into two sets of sizes differing by at most $1$, and connects each vertex in first set to $v_1$ and each vertex in the other set to $v_2$.
\item Upon seeing a constraint that does not start with $v_1$ and $v_2$, which reveals the center of the star, it connects the center to all vertices that are not yet connected to the center. 
\end{enumerate}
Since the algorithm adds, at most $n/2-1$ edges to the wrong center, this gives us an asymptotic upper bound $3/2$, which matches the lower bound.
\qed
\end{proof}

\subsection{Paths}\label{subsubsec:paths}

In the next two theorems, we give matching lower and upper bounds (in the limit) for path graphs.

\begin{theorem}\label{thm:onlinePathLower}
The competitive ratio for the \textbf{Online \NCOC} problem when the algorithm knows that the optimal solution forms a \textbf{path} has an asymptotic lower bound of $2$.
\end{theorem}
\begin{proof}
Fix an arbitrary ordering of the vertices $\{v_1, v_2, v_3,\dots, v_{n}\}$. For $3\leq i\leq n$, define the \textbf{pre-degree} of a vertex $v_i$ to be the number of neighbors $v_i$ has in $\{v_1, v_2, v_3,\dots, v_{i-1}\}$. Algorithm~\ref{alg:starLowerBound} below is a simple strategy the adversary can take to force $v_3,\dots, v_n$ to all have pre-degree at least $2$. Since any algorithm will add at least $2n-3$ edges, this gives an asymptotic lower bound of $2$. 
\vspace{-.5cm}
\begin{algorithm}[!ht]
\caption{{Forcing pre-degree to be at least $2$}}
\label{alg:starLowerBound}
	\begin{algorithmic}
	\STATE Give ordered constraint $\mathcal{O}=(v_1, v_2, v_3,\dots, v_n)$ to the algorithm;
	\FOR {$i=3$ to $n$} 
		\IF {the pre-degree of $v_i$ is at least $2$} \STATE continue;
		\ELSE
			\STATE pick up at random a path (say $P_{i}$) that satisfies all the constraints up to this round and an endpoint $u$ of the path that is not connected to $v_i$, and gives the algorithm the constraint $\paren{v_i, u}$;
		\ENDIF
	\ENDFOR
	\end{algorithmic}
\end{algorithm}
\vspace{-.5cm}
Suppose $P_{i}$ was the path picked in round $i$ (i.e. $P_i$ satisfies all constraints upto round $i$). Then, $P_i$ along with the edge $\paren{v_i, u}$ is a path that satisfies all constraints upto round $i+1$. Hence by induction, for all $i$, there is a path that satisfies all constraints given by the adversary upto round $i$.
\qed
\end{proof}

\begin{theorem}\label{thm:onlinePathUpper}
The competitive ratio for the \textbf{Online \NCOC} problem when the algorithm knows that the optimal solution forms a \textbf{path} has an asymptotic upper bound of $2$.
\end{theorem}

\begin{proof}
For our algorithm matching the upper bound, we use the pq-trees, introduced by Booth and Lueker~\cite{boothL1976testing}, which keep track all consistent permutations of vertices given contiguous intervals of vertices.
Our analysis is based on ideas from Angluin~et~al.~\cite{angluinAR2015network}, who also 
use pq-trees for analyzing the general problem.\footnote{Angluin~et~al.~\cite{angluinAR2015network} have a small error in their argument because their potential function fails to explicitly consider the number of p-nodes, which creates a problem for some of the pq-tree updates. We fix this, without affecting their asymptotic bound. For the ordered constraints case, we are also able to obtain a much finer analysis.}

 A \textbf{pq-tree} is a tree whose leaf nodes are the vertices and each internal node is either a \textbf{p-node} or a \textbf{q-node}.
\begin{itemize}
\item A \textbf{p-node} has a two or more children of any type. The children of a p-node form a contiguous interval that can be in any order.
\item A \textbf{q-node} has three or more children of any type. The children of a q-node form a contiguous interval, but can only be in the given order of its inverse.
\end{itemize}
Every time a new interval constraint comes, the tree update itself by identifying any of the eleven patterns, P0, P1,$\dots$, P6, and Q0, Q1, Q2, Q3, of the arrangement of nodes and replacing it with each correspondent replacement. The update fails when it cannot identify any of the  patterns, in which case the contiguous intervals fail to produce any consistent permutation. We refer readers to Section 2 of Booth and Lueker~\cite{boothL1976testing} for a more detailed description of pq-trees.

The reason we can use a pq-tree to guide our algorithm is because of an observation made in Section~\ref{sec:introduction} that each ordered constraint $\paren{v_1,v_2, v_3,\dots, v_{k-1}, v_k}$ is equivalent to $k$ interval constraints $\set{v_1,v_2}, \set{v_1,v_2,v_3},\cdots,\set{v_1,\dots,v_{k-1}},$ $\set{v_1,\dots,v_{k-1}, v_k}.$ So upon seeing one ordered constraints, we reduce the pq-tree with the equivalent interval constraints, \emph{in order}. Then what our algorithm does is simply to add edge(s) to the graph every time a pattern is identified and replaced with its replacement, so that the graph satisfies all the seen constraints. Note that to reduce the pq-tree with one interval constraint, there may be multiple patterns identified and hence multiple edges may be added. 

Before running into details of how the patterns determine which edge(s) to add, we note that, without loss of generality, we can assume that the the algorithm is in either one of the following two stages.
\begin{itemize}
\item The pq-tree is about to be reduced with $\set{v_1,v_2}$.
\item The pq-tree is about to be reduced with $\set{v_1,\dots, v_k}$, when the reductions with $\set{v_1,v_2},\cdots$, $\set{v_1,\dots,v_{k-1}}$ have been done.
\end{itemize}

Because of the structure of constraints discussed above, we do not encounter all pq-tree patterns in their full generality, but in the special forms demonstrated in Table~\ref{tab:pq-treePatternReplacement}. Based on this, we make three important observations which can be verified by carefully examining how a pq-tree evolves along with our algorithm.
\begin{enumerate}
\item The only p-node that can have more than two children is the root.
\item At least one of the two children of a non-root p-node is a leaf node.
\item For all q-nodes, there must at least one leaf node in any two adjacent children. Hence, Q3 doesn't appear.
\end{enumerate}

Now we describe how the edges are going to be added. Note that a pq-tree inherently learns edges that appear in optimum even when those edges are not forced by constraints. Apart from adding edges that are necessary to satisfy the constraints, our algorithm will also add any edge that the pq-tree inherently learns. For all the patterns except Q2 such that a leaf node $v_k$ is about to be added as a child to a different node, we can add one edge joining $v_k$ to $v_{k-1}$. For all such patterns except Q2, it is obvious that this would satisfy the current constraint and all inherently learnt edges are also added. For Q2, the pq-tree could learn two edges. The first edge is $(v_k,v_{k-1})$. The second one is an edge between the leftmost subtree of the daughter q-node (call $T_l$) and the node to its left (call $v_l$). Based on Observation 3, $v_l$ is a leaf. But based on the algorithm, one of these two edges is already added. Hence, we only need to add one edge when Q2 is applied. For P5, we add the edge as shown in Table~\ref{tab:pq-treePatternReplacement}.

\begin{small}
\begin{center}
\begin{longtable}{m{.9cm}| >{\centering\arraybackslash} m{5.4cm} >{\centering\arraybackslash} m{5.4cm}}
&Pattern & Replacement\\ \hline
P2& \includegraphics[scale=.5]{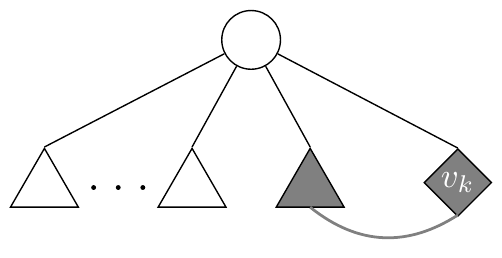} & \includegraphics[scale=.5]{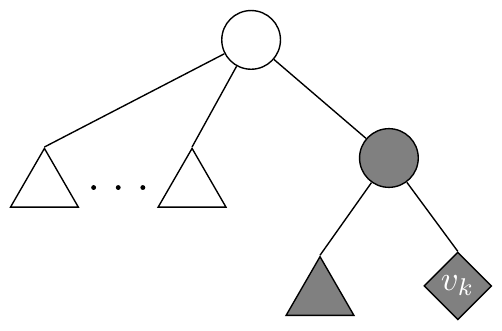}\\
P3&\includegraphics[scale=.5]{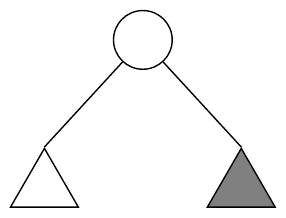} & \includegraphics[scale=.5]{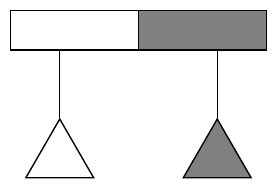}\\
P4(1)&\includegraphics[scale=.5]{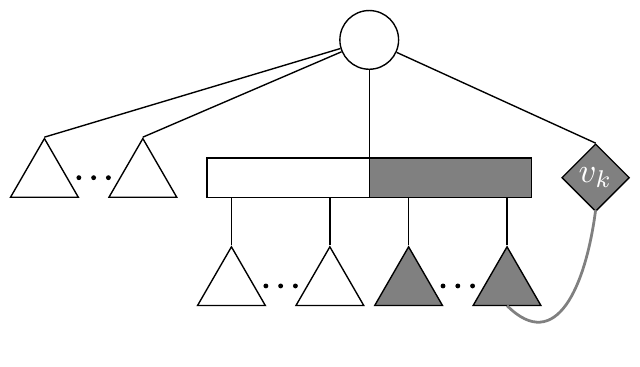} & \includegraphics[scale=.5]{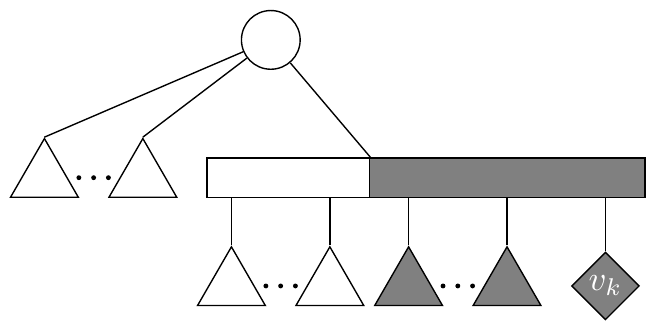}\\
P4(2)&\includegraphics[scale=.5]{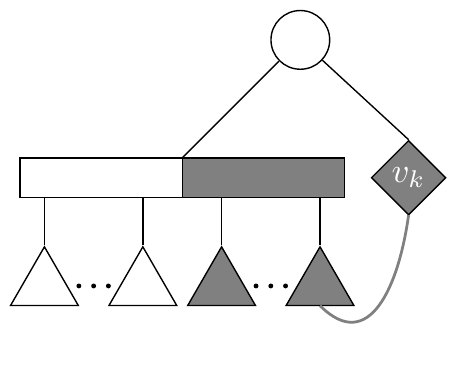} & \includegraphics[scale=.5]{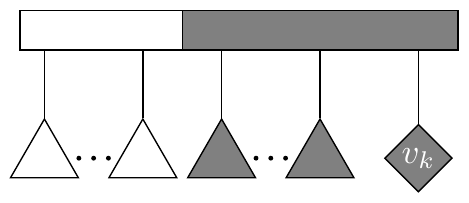}\\
P5&\includegraphics[scale=.5]{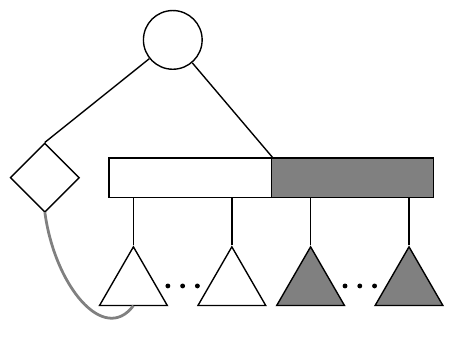} & \includegraphics[scale=.5]{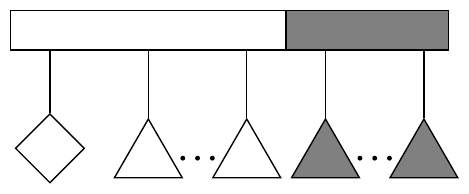}\\
P6(1)&\includegraphics[scale=.5]{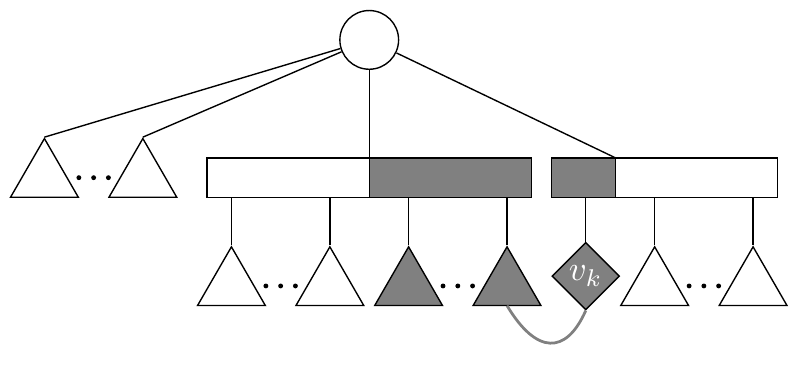} & \includegraphics[scale=.5]{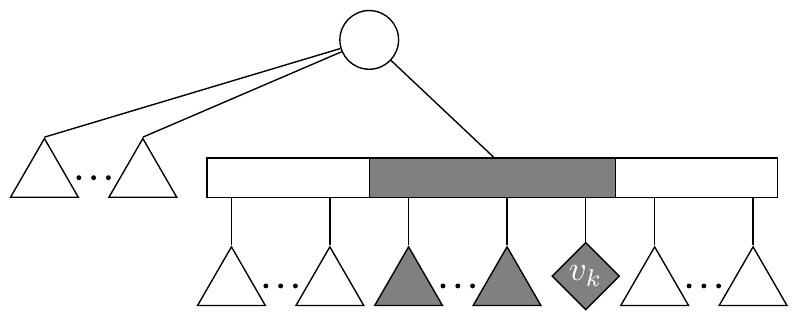}\\
P6(2)&\includegraphics[scale=.5]{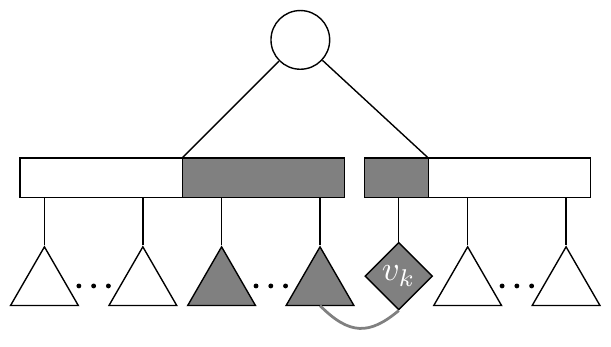} & \includegraphics[scale=.5]{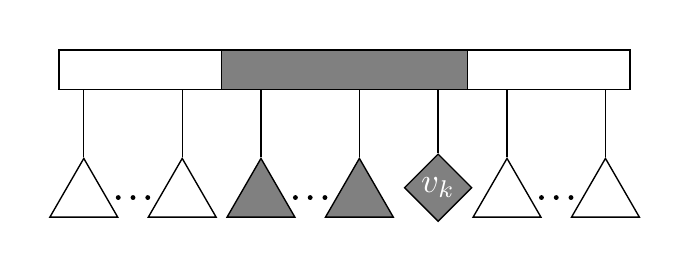}\\
Q2& \includegraphics[scale=.5]{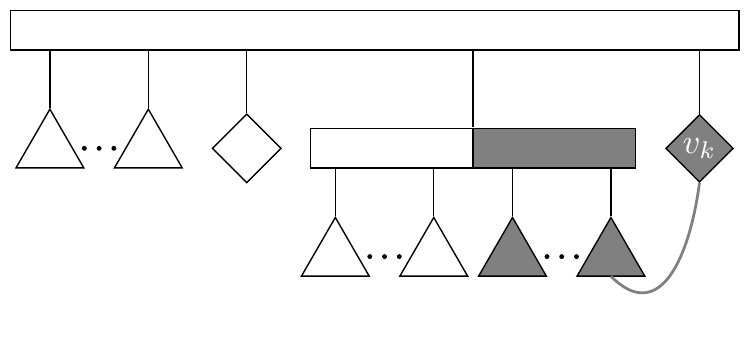} & \includegraphics[scale=.5]{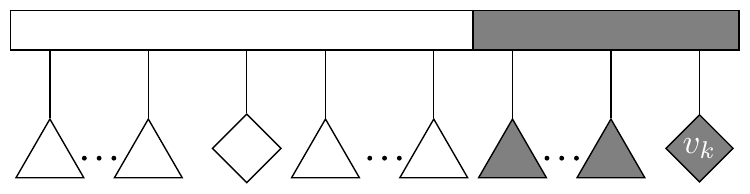}\\
\hline
\caption{Specific patterns and replacements that appear through the algorithm. P4(1) denotes the case of P4 where the top p-node is retained in the replacement and P4(2) denotes the case where the top p-node is deleted. The same is true for P6. P0, P1, Q0, and Q1 are just relabelling rules, and we have omitted them because no edges need to be added. We use the same shapes to represent p-nodes, q-nodes, and subtrees as in Booth and Lueker's paper~\cite{boothL1976testing} for easy reference, and we use diamonds to represent leaf nodes. }
\label{tab:pq-treePatternReplacement}
\end{longtable}
\end{center}
\end{small}

Let us denote by $P$ and $Q$ the sets of p-nodes and q-nodes, respectively, and by $c(p)$ the number of children node $p$ has. And let potential function $\phi$ of a tree $T$ be defined as
\[
\phi(T)=a\sum_{p\in P}c(p)+b|P|+c|Q|,
\]
where $a$, $b$, and $c$ are coefficients to be determined later. 

\begin{table}[ht]
\centering
\begin{tabular}{l|rrr|c|c}
& $\sum_{p\in {P}}c(p)$ & $|P|$ & $|Q|$& $-\Delta \Phi$& number of edges added\\ \hline
P2 & $1$ & $1$ & $0$ & $-a-b$& $1$\\
P3 & $-2$ & $-1$ & $1$ & $2a+b-c$ & $0$\\
P4(1) & $-1$ & $0$ & $0$ & $a$ & $1$\\
P4(2) & $-2$ & $-1$ & $0$ & $2a+b$ & $1$\\
P5 & $-2$ & $-1$ & $0$ & $2a+b$ & $1$\\
P6(1) & $-1$ & $0$ & $-1$ & $a+c$ & $1$\\
P6(2) & $-2$ & $-1$ & $-1$ & $2a+b+c$ & $1$\\
Q2 & $0$ & $0$ & $-1$ & $c$ & $1$\\
Q3 & $0$ & $0$ & $-2$ & $2c$ & $1$\\ \hline
\end{tabular}
\vskip .1in
\caption{How the terms in the potential function: $\sum_{p\in {P}}c(p)$, $|P|$, and $|Q|$ change according to the updates.}
\label{tab:pq-treePotentialFunc}
\vskip -.3in
\end{table}

We want to upper bound the number of edges added for each pattern by the drop of potential function. We collect the change in the three terms in the potential function that each replacement causes in Table~\ref{tab:pq-treePotentialFunc}, and we can solve a simple linear system to get that choosing $a=2$, $b=-3$, and $c=1$ is sufficient. For ease of analysis, we add a dummy vertex $v_{n+1}$ that does not appear in any constraint. Now, the potential function starts at $2n-1$ (a single p-node with $n+1$ children) and decreases to $2$ when a path is uniquely determined. Hence, the number of edges added
by the algorithm is $2n-3$, which gives the desired asymptotic upper bound.
\qed
\end{proof}



%
%
\bibliographystyle{plain}
\bibliography{paper}
\end{document}